\documentclass[runningheads]{llncs}
\usepackage[a4paper,margin=3cm]{geometry}
\usepackage{microtype,verbatim,dsfont}
\usepackage[utf8]{inputenc}
\usepackage[inline]{enumitem}
\usepackage{xcolor,soul}
\usepackage{lineno,xspace}
\usepackage{todonotes}
\usepackage[labelfont=bf]{caption}
\usepackage{subcaption}
\definecolor{linksblue}{rgb}{0.0,0.1,0.6}

\usepackage{amsmath,amssymb,amsthm}

\usepackage[pdfpagelabels,colorlinks,citecolor=linksblue,linkcolor=linksblue,urlcolor=linksblue]{hyperref}
\usepackage[nameinlink,capitalize]{cleveref}
\Crefformat{figure}{#2Fig.~#1#3}

\renewcommand{\emph}[1]{\textcolor{linksblue}{\it #1}\xspace}
\newcommand{\bound}{42\xspace}

\renewcommand{\orcidID}[1]{\href{https://orcid.org/#1}{\includegraphics[scale=.03]{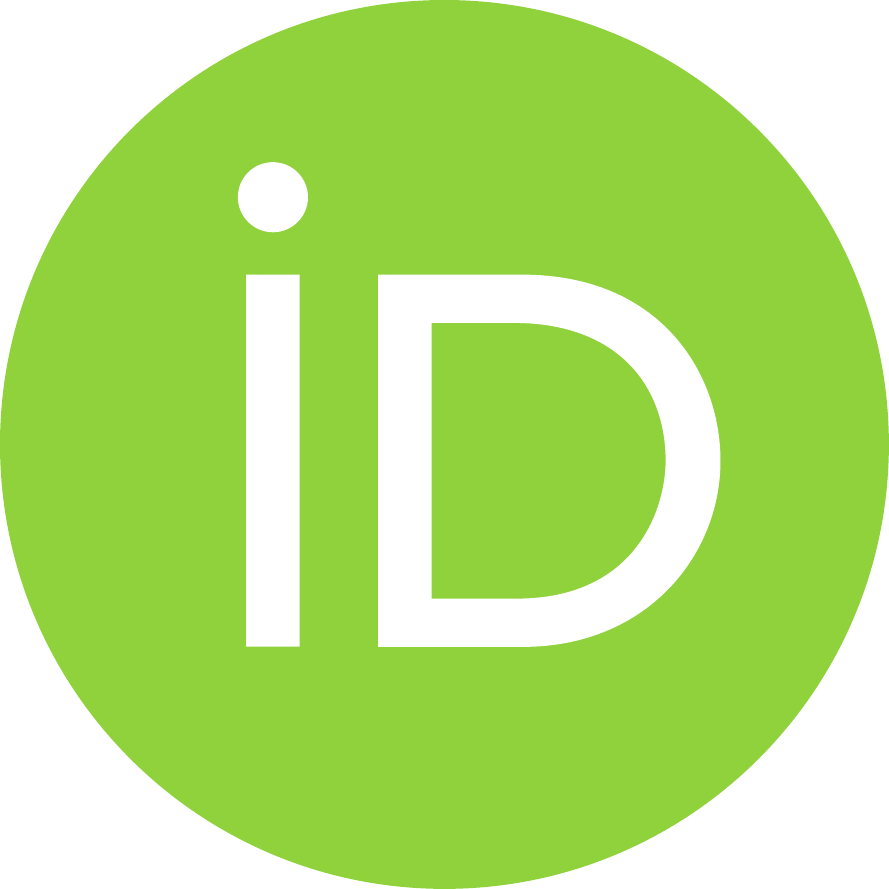}}}

\newcommand{\Qh}{{\ensuremath{\mathcal{Q}}}}
\newcommand{\lo}[1]{{\ensuremath{\prec_\mathrm{\tiny #1}}}}
\newcommand{\ql}[1]{\ensuremath{\mathit{QL}(#1)}}

\newenvironment{myenumerate}[1]
{\begin{enumerate}[label={#1},labelindent=0.0\parindent,leftmargin=*,align=left]}{\end{enumerate}}
 
\newenvironment{myenumerate2}[1]
{\begin{enumerate}[label={#1},leftmargin=30pt,align=right]
\setlength{\labelwidth}{30pt}
\setlength{\itemindent}{0.0em}
\setlength{\labelsep}{1ex} 
}{\end{enumerate}}

\begin{document}

\title{On the Queue Number of Planar Graphs}
\titlerunning{On the Queue Number of Planar Graphs}
\author{Michael~A.~Bekos\inst{1}\orcidID{0000-0002-3414-7444}, Martin~Gronemann\inst{2}\orcidID{0000-0003-2565-090X}, Chrysanthi~N.~Raftopoulou\inst{3}\orcidID{0000-0001-6457-516X}}

\institute{%
Department of Computer Science, University of T{\"u}bingen, T{\"u}bingen, Germany
\\\email{bekos@informatik.uni-tuebingen.de}
\and
Theoretical Computer Science, Osnabr\"uck University, Osnabr\"uck, Germany
\\\email{martin.gronemann@uni-osnabrueck.de}
\and
School of Applied Mathematics and Physical Sciences, NTUA, Athens, Greece
\\\email{crisraft@math.ntua.gr}
}

\maketitle

\begin{abstract}
A $k$-queue layout is a special type of a linear layout, in which the linear order avoids $(k+1)$-rainbows, i.e., $k+1$ independent edges that pairwise form a nested pair. The optimization goal is to determine the \emph{queue number} of a graph, i.e., the minimum value of $k$ for which a $k$-queue layout is feasible.
Recently, Dujmović et al.~[J.\ ACM, 67(4), 22:1-38, 2020] showed that the queue number of planar graphs is at most $49$, thus settling in the positive a long-standing conjecture by Heath, Leighton and Rosenberg. To achieve this breakthrough result, their approach involves three different techniques: 
(i)~an algorithm to obtain straight-line drawings of outerplanar graphs, in which the $y$-distance of any two adjacent vertices is $1$ or $2$,
(ii)~an algorithm to obtain $5$-queue layouts of planar $3$-trees, and 
(iii)~a decomposition of a planar graph into so-called tripods.
In this work, we push further each of these techniques to obtain the first non-trivial improvement of the upper bound on queue number of planar graphs from $49$ to $\bound$.
\end{abstract}

\section{Introduction}
\label{sec:introduction}

Linear layouts of graphs have a long tradition of study in different contexts, including graph theory and graph drawing, as they form a framework for defining different graph-theoretic parameters with several applications; see, e.g.,~\cite{DBLP:journals/csur/DiazPS02}. In this regard, one seeks to find a total order of the vertices of a graph that reaches a certain optimization goal~\cite{DBLP:journals/ita/BarthPRR95,DBLP:journals/jgt/ChinnCDG82,DBLP:journals/dam/HortonPB00}. In this work, we focus on a well-studied type of linear layouts, called \emph{queue layout}~\cite{DBLP:journals/jacm/DujmovicJMMUW20,DBLP:journals/siamdm/HeathLR92,DBLP:journals/siamcomp/HeathR92,DBLP:journals/combinatorics/Wiechert17}, in which the goal is to minimize the size of the largest \emph{rainbow}, namely, a set of independent edges that are pairwise nested. Equivalently, the problem asks for a linear order of the vertices~and a partition of the edges into a minimum number of queues (called \emph{queue~number}), such that no two independent edges in the same queue are nested~\cite{DBLP:journals/siamcomp/HeathR92}; see \cref{fig:queue}.

Queue layouts of graphs were introduced by Heath and Rosenberg~\cite{DBLP:journals/siamcomp/HeathR92} in~1992 as the counterpart of \emph{stack layouts} (widely known also as \emph{book embeddings}), in which the edges must be partitioned into a minimum number of stacks (called \emph{stack~number}), such that no two edges in the same stack cross~\cite{DBLP:journals/jct/BernhartK79}; see \cref{fig:stack}. Since their introduction, queue layouts of graphs have~been a fruitful subject~of~intense research with several important milestones over the years~\cite{DBLP:journals/algorithmica/AlamBGKP20,DBLP:journals/algorithmica/BannisterDDEW19,DBLP:journals/siamcomp/BekosFGMMRU19,DBLP:journals/siamcomp/DujmovicMW05,DBLP:journals/siamdm/HeathLR92,DBLP:journals/combinatorics/Wiechert17,DBLP:conf/fsttcs/Wood02}; for an introduction, we refer the interested reader to~\cite{DBLP:journals/dmtcs/DujmovicW04}. 

The most intriguing problem in this research field is undoubtedly the problem of specifying the queue number of planar graphs, that is, the maximum queue number of a planar graph. This problem dates back to a conjecture by Heath, Leighton and Rosenberg, who in 1992 conjectured that the queue number of planar graphs is bounded~\cite{DBLP:journals/siamdm/HeathLR92}. Notably, despite the different efforts~\cite{DBLP:journals/algorithmica/BannisterDDEW19,DBLP:journals/siamcomp/BattistaFP13,DBLP:journals/jgaa/DujmovicF18}, this conjecture remained unanswered for more than two decades. 
That only changed in 2019 with a breakthrough result of Dujmović, Joret, Micek, Morin, Ueckerdt and Wood~\cite{DBLP:conf/focs/DujmovicJMMUW19}, who managed to settle in the positive the conjecture, as they showed that the queue~number of planar graphs is at most $49$. The best-known corresponding lower bound is~$4$ due to Alam et al.~\cite{DBLP:journals/algorithmica/AlamBGKP20}. 

It is immediate to see, however, that the gap between the currently best-known lower and upper bounds is rather large, which implies that the exact queue~number of planar graphs is, up to the point of writing, still unknown. Note that this is in contrast with the maximum stack~number of planar graphs, which was recently shown to be exactly $4$~\cite{DBLP:journals/jocg/KaufmannBKPRU20,DBLP:journals/jcss/Yannakakis89}. Also, the existing gap in the bounds on the queue number of planar graphs gives the intuition that it is unlikely the upper bound of $49$ by Dujmović at al.~\cite{DBLP:conf/focs/DujmovicJMMUW19} to be tight, even though in the last two years no improvement appeared in the literature.

\begin{figure}[t!]
	\centering	
	\begin{subfigure}[b]{.25\textwidth}
		\centering
		\includegraphics[page=3]{figures/octahedron}
		\caption{}
		\label{fig:goldner-harary}
	\end{subfigure}
	\hfil
	\begin{subfigure}[b]{.35\textwidth}
		\centering
		\includegraphics[page=2]{figures/octahedron}
		\caption{2-queue layout}
		\label{fig:queue}
	\end{subfigure}
	\hfil
	\begin{subfigure}[b]{.35\textwidth}
		\centering
		\includegraphics[page=1]{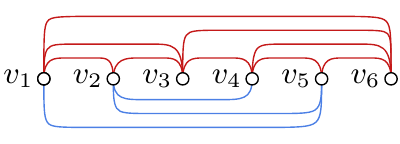}
		\caption{2-stack layout}
		\label{fig:stack}
	\end{subfigure}
	\caption{
    (a)~The octahedron graph and 
    (b)-(c)~different linear layouts of it.}
	\label{fig:sample}
\end{figure}

\paragraph{Our contribution.} We verify the aforementioned intuition by reducing the upper bound on the queue number of planar graphs from $49$ to $\bound$ (see \Cref{thm:main} in \Cref{sec:main}). To achieve this, we present improvements to each of the following three main techniques involved in the approach by Dujmović at al.~\cite{DBLP:conf/focs/DujmovicJMMUW19}: 
\begin{enumerate*}[label={(\roman*)}]
    \item an algorithm to obtain straight-line drawings of outerplanar graphs, in which the $y$-distance of any two adjacent vertices is $1$~or~$2$~\cite{DBLP:journals/dmtcs/DujmovicPW04}; see \Cref{sec:outerplanar-refine}, 
    \item an algorithm to obtain $5$-queue layouts of planar $3$-trees~\cite{DBLP:journals/algorithmica/AlamBGKP20}; see \cref{sec:planar-3-trees-refine}, and 
    \item a decomposition of a planar graph into so-called tripods~\cite{DBLP:conf/focs/DujmovicJMMUW19}; see \cref{sec:planar-refine}. 
\end{enumerate*}
Although we assume familiarity with these techniques, we outline in \cref{sec:all-outlines} their most important~aspects; also, for preliminary notions and standard terminology refer to \cref{sec:preliminaries}. 

\section{Preliminaries}
\label{sec:preliminaries}

A \emph{drawing} of a graph maps each vertex to a distinct point of the Euclidean plane and each edge to a Jordan curve connecting its endpoints. A drawing of a graph is \emph{planar}, if no two edges cross (expect at common endpoints). A graph that admits a planar drawing is called \emph{planar}. A planar drawing partitions the plane into topologically connected regions called \emph{faces}; its unbounded one is called \emph{outer}. A \emph{combinatorial embedding} of a planar graph is an equivalence class of planar drawings that are pairwise topologically equivalent, i.e., they define the same set of faces up to an orientation-preserving homeomorphism of the plane. A planar graph  together with a combinatorial embedding is a \emph{plane graph}. A graph that admits a planar drawing, in which all vertices are incident to its outer face, is called \emph{outerplanar}. It is known that outerplanar graphs have treewidth at most~$2$. A planar graph with treewidth at most $3$ is commonly referred to as a \emph{planar $3$-tree}. Planar $3$-trees are subgraphs of \emph{maximal planar $3$-trees}, where such an $n$-vertex graph is either a $3$-cycle, if $n=3$, or has a degree-$3$ vertex whose deletion yields a maximal planar $3$-tree with $n-1$ vertices, if $n > 3$.

A \emph{vertex order} $\prec$ of a simple undirected graph $G$ is a total order of its vertices, such that for any two vertices $u$ and $v$ of $G$, $u \prec v$ if and only if $u$ precedes $v$ in the order. 
 Let $F$ be a set of $k \geq 2$ independent edges $(u_i, v_i)$ of $G$, that is, $F=\{(u_i, v_i);\;i=1,\ldots,k\}$. 
If $u_1\prec \ldots\prec u_k\prec v_k\prec \ldots\prec v_1$, then we say that the edges of $F$ form a \emph{$k$-rainbow}, while if $u_1\prec v_1\prec \ldots\prec u_k\prec v_k$, then the edges of $F$ form a \emph{$k$-necklace}. The edges of $F$ form a \emph{$k$-twist}, if $u_1\prec \ldots\prec u_k\prec v_1\prec \ldots\prec v_k$; see \cref{fig:necklace-twist}. Two independent edges that form a $2$-rainbow ($2$-necklace, $2$-twist) are referred to as \emph{nested} (\emph{disjoint}, \emph{crossing}, respectively).

\begin{figure}[h]
	\centering	
	\begin{subfigure}[b]{.32\textwidth}
		\centering
		\includegraphics[scale=1,page=1]{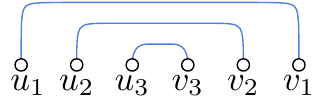}
		\caption{}
		\label{fig:rainbow}
	\end{subfigure}
	\hfil
	\begin{subfigure}[b]{.32\textwidth}
		\centering
		\includegraphics[scale=1,page=3]{figures/preliminaries}
		\caption{}
		\label{fig:necklace}
	\end{subfigure}
	\hfil
	\begin{subfigure}[b]{.32\textwidth}
		\centering
		\includegraphics[scale=1,page=2]{figures/preliminaries}
		\caption{}
		\label{fig:twist}
	\end{subfigure}
	\caption{Illustration of
		(a)~a $3$-rainbow, 
		(b)~a $3$-necklace, and 
		(c)~a $3$-twist.}
	\label{fig:necklace-twist}
\end{figure}

A \emph{$k$-queue layout} of a graph consists of a vertex order $\prec$ of $G$ and a partition of the edge-set of $G$ into $k$ sets of pairwise non-nested edges, called \emph{queues}. A preliminary result by Heath and Rosenberg~\cite{DBLP:journals/siamcomp/HeathR92} states that a graph admits a $k$-queue layout if and only if it admits a vertex order in which no $(k+1)$-rainbow is formed. The \emph{queue number} of a graph $G$, denoted by $\mathrm{qn}(G)$, is the minimum
$k$, such that $G$ admits a $k$-queue layout. Accordingly, the queue number of a class of graphs is the maximum queue number over all its members.

\section{Sketch of the involved techniques}
\label{sec:all-outlines}

In the following, we sketch the main aspects of the three algorithms mentioned in \cref{sec:introduction} that are involved in the approach by Dujmović et al.~\cite{DBLP:journals/jacm/DujmovicJMMUW20} to achieve the bound of $49$ on the queue number of planar graphs. 
The first is an algorithm by Dujmović, P{\'o}r, and Wood~\cite{DBLP:journals/dmtcs/DujmovicPW04} to compute $2$-queue layouts of outerplanar graphs (\cref{sec:outerplanar-outline}). The second one is by Alam et al.~\cite{DBLP:journals/algorithmica/AlamBGKP20} to compute $5$-queue layouts of planar $3$-trees (\cref{sec:planar-3-trees-outline}). The last one is the actual algorithm by Dujmović et al.~\cite{DBLP:journals/jacm/DujmovicJMMUW20} to compute $49$-queue layouts of planar graphs (\cref{sec:planar-outline}).

\subsection{Outerplanar Graphs}
\label{sec:outerplanar-outline}

The main ingredient of the algorithm by Dujmović, P{\'o}r and Wood~\cite{DBLP:journals/dmtcs/DujmovicPW04}  is an algorithm to obtain a straight-line drawing $\Gamma(G)$ of a maximal outerplane graph $G$ whose output can be transformed into a $2$-queue layout of $G$.
The recursive construction of $\Gamma(G)$ maintains the following invariant~properties:%
\begin{myenumerate}{(O.\arabic*)}
\item\label{outer:1} The cycle delimiting the outerface consists of two strictly $x$-monotone paths, referred to as \emph{upper} and \emph{lower envelopes}, respectively.
\item\label{outer:2} The $y$-coordinates of the endvertices of each edge differ by either one (\emph{span-$1$} edge) or two (\emph{span-$2$} edge). 
\end{myenumerate}

\begin{figure}[t]
   \centering
   \begin{subfigure}{0.48\textwidth}
        \centering
        \includegraphics[page=1]{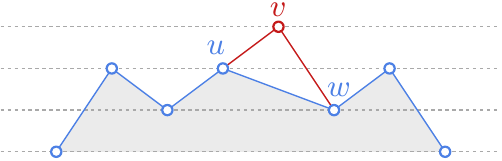}
        \caption{}
    \label{fig:outerplanar_augment_1}
    \end{subfigure}
    \hfil
    \begin{subfigure}{0.48\textwidth}
        \centering
        \includegraphics[page=2]{figures/outerplanar.pdf}
        \caption{}
        \label{fig:outerplanar_augment_2}
    \end{subfigure}
    \caption{Introducing a degree-$2$ vertex $v$ along the upper envelope, when its two neighbors $u$ and $w$ are connected with (a)~a span-1 edge, and
    (b)~a span-2 edge.}
    \label{fig:outerplanar}
\end{figure}

\noindent To maintain \ref{outer:1} and~\ref{outer:2}, Dujmović et al.\ adopt an approach in which at each recursive step a vertex of degree $2$ is added to the already constructed drawing; see \cref{fig:outerplanar}.
The base of the recursion consists of a triangle\footnote{In \cite{DBLP:journals/dmtcs/DujmovicPW04}, Dujmović et al.\ consider a single edge at the base of the recursion. To ease the presentation, we slightly modify their approach.} that can be trivially drawn with two span-$1$ edges and one span-$2$ edge satisfying Invariants~\ref{outer:1} and~\ref{outer:2}. Assume that $G$ has $n > 3$ vertices. Since $G$ is biconnected outerplane, it contains a vertex $v$ of degree $2$. Removing $v$ yields a biconnected outerplane graph $G'$ with $n-1$ vertices, which recursively admits a drawing $\Gamma(G')$ satisfying Invariants~\ref{outer:1} and \ref{outer:2}. By Invariant~\ref{outer:1}, none of the edges in $\Gamma(G')$ is drawn vertically. To obtain drawing $\Gamma(G)$ of $G$ vertex $v$ is introduced in $\Gamma(G')$ as follows. Let $u$ and $w$ be the neighbors of $v$ in $G$. Since $G$ is maximal outerplane, $(u,w)$ is an edge of $G$ that lies on the outerface of $\Gamma(G')$. By Invariant~\ref{outer:2}, $(u,w)$ is either a span-$1$ or a span-$2$ edge. Assume that $(u,w)$ is along the upper envelope of $\Gamma(G')$; the case where it is along the lower envelope is symmetric. Assume first that $(u,w)$ has span $1$ in $\Gamma(G')$ and w.l.o.g.\ that $y(u)=y(w)+1$. Then, vertex $v$ is placed such that $y(v)=y(u)+1$ and $x(v)=\frac{1}{2}(x(u)+x(w))$; see \cref{fig:outerplanar_augment_1}. Hence, $(u,v)$ and $(v,w)$ have span $1$ and $2$, respectively, which implies that Invariant~\ref{outer:2} is maintained. Since edge $(u,w)$ of the upper envelope of $\Gamma(G')$ is replaced by the $x$-monotone path $u \rightarrow v \rightarrow w$ in $\Gamma(G)$, Invariant~\ref{outer:1} is also maintained. Assume now that $(u,w)$ has span $2$ in $\Gamma(G')$ and w.l.o.g. that $y(u)=y(w)+2$. Then, vertex $v$ is placed such that $y(v)=y(u)-1$ and $x(v)=\frac{3x(u)+x(w)}{4}$; see \cref{fig:outerplanar_augment_2}. This implies that both edges $(u,v)$ and $(v,w)$ have span $1$. Similarly to the previous case, one can argue that Invariants~\ref{outer:1} and~\ref{outer:2} are maintained.
Drawing $\Gamma(G)$ is transformed to a $2$-queue layout of $G$ as follows:
\begin{enumerate}[label=(\roman*)]
\item for any two vertices $u$ and $v$ of $G$, $u \prec v$ if and only if either $y(u) > y(v)$, or $y(u)=y(v)$ and $x(u)<x(v)$ in $\Gamma(G)$, 
\item edge $(u,v)$ is assigned to the first (second) queue if it has span $1$ ($2$, respectively)~in~$\Gamma(G)$.
\end{enumerate}

\subsection{Planar 3-trees}
\label{sec:planar-3-trees-outline}

Alam et al.~\cite{DBLP:journals/algorithmica/AlamBGKP20} adopt a \emph{peeling-into-levels} approach~\cite{DBLP:conf/focs/Heath84}~to~produce a $5$-queue layout of a maximal plane $3$-tree $H$. Initially, the vertices of $H$ are partitioned into \emph{levels} $L_0,\ldots,L_\lambda$ with $\lambda \geq 1$, such that $L_0$-vertices are incident to the outer face of $H$, while $L_{i+1}$-vertices are in the outer face of the subgraph of $H$ obtained by the removal of all vertices in $L_0,\ldots,L_i$. The edges of $H$ are partitioned into \emph{level} and \emph{binding}, depending on whether their endpoints are on the same or on consecutive levels. As each connected component of the subgraph $H_i$ of $H$ induced by the edges of level $L_i$ is an internally triangulated outerplane graph, it is embeddable in two queues. This implies that each connected component $c$ of $H_{i+1}$ (which is outerplane) lies in the interior of a triangular face of $H_i$, therefore there are exactly three vertices of $H_i$ that are connected to $c$.
The constructed $5$-queue layout of $H$ satisfies the following invariant properties:
\begin{myenumerate}{(T.\arabic*)}
\item \label{t:order} The linear order $\lo{H}$ is such that all vertices of level $L_j$ precede all vertices of level $L_{j+1}$ for every $j=0,\ldots,\lambda-1$;
\item \label{t:order_component} Vertices of each connected component of level $L_j$ appear consecutively in $\lo{H}$ for every $j=0,\ldots,\lambda$;
\item \label{t:level} Level edges of each of the levels $L_{0},\ldots,L_\lambda$ are assigned to two queues denoted by $\Qh_0$ and $\Qh_1$;
\item \label{t:bind} For every $j=0,\ldots,\lambda-1$, the binding edges between $L_j$ and $L_{j+1}$ are~assigned to three queues $\Qh_2$, $\Qh_3$ and $\Qh_4$ as follows. For each connected~component $c$ of $H_{j+1}$, let $x$, $y$ and $z$ be its three neighbors in $H_j$ so that $x\;\lo{H}\; y\;\lo{H}\; z$.
Then, the binding edges between $L_j$ and $L_{j+1}$ incident to $c$ are assigned to $\Qh_2$, $\Qh_3$ and $\Qh_4$ if they lead to $x$, $y$ and $z$, respectively.
\end{myenumerate}

\subsection{General Planar Graphs}
\label{sec:planar-outline}

Central in the algorithm by Dujmović et al.~\cite{DBLP:journals/jacm/DujmovicJMMUW20} is the notion of $H$-partition\footnote{To avoid confusion with notation used earlier, note that, in the scope~of~the algorithm by Dujmović et al.~\cite{DBLP:journals/jacm/DujmovicJMMUW20}, graph $H$ denotes a plane $3$-tree, as we will~shortly~see.}, defined as follows. Given a graph $G$, an \emph{$H$-partition} of $G$ is a partition of the vertices of $G$ into sets $A_x$ with $x\in V(H)$, called \emph{bags}, such that for each edge $(u,v)$ of $G$ with $u\in A_x$ and $v \in A_y$ either $x=y$ holds or $(x,y)$ is an edge of $H$. In the former case, $(u,v)$ is called \emph{intra-bag} edge, while in the latter case \emph{inter-bag}.
A \emph{BFS-layering} of $G$ is a partition $\mathcal{L}=(V_0,V_1,\ldots)$ of its vertices according to their distance from a specific vertex $r$ of $G$, 
i.e., it is a special type of $H$-partition, where $H$ is a path and each bag $V_i$ corresponds to a \emph{layer}. In this regard, an intra-bag edge is called \emph{intra-layer}, while an inter-bag edge is called \emph{inter-layer}\footnote{Dujmović et al.~\cite{DBLP:journals/jacm/DujmovicJMMUW20} refer to the intra- and  inter-\emph{layer} edges as intra- and inter-\emph{level} edges, respectively. We adopt the terms intra- and inter-layer edges to avoid confusion with the different type of leveling used in the algorithm of Alam et al.~\cite{DBLP:journals/algorithmica/AlamBGKP20}.}. An $H$-partition has \emph{layered-width} $\ell$ with respect to a BFS-layering  $\mathcal{L}$ if each bag of $H$ has at most $\ell$ vertices on each layer of $\mathcal{L}$. 

\begin{figure}[t]
   \centering
   \begin{subfigure}{0.45\textwidth}
        \centering
        \includegraphics[page=1]{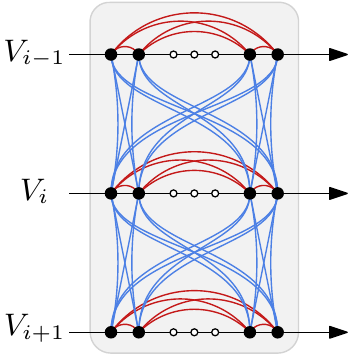}
        \caption{Intra-bag edges}
    \label{fig:intra_bag_edges}
    \end{subfigure}
    \hfil
    \begin{subfigure}{0.54\textwidth}
        \centering
        \includegraphics[page=2]{figures/planar_queues_fig_edges1}
        \caption{Inter-bag edges}
        \label{fig:inter_bag_edges}
    \end{subfigure}
    \caption{Illustration of
    (a)~Intra-bag edges; the intra-layer ones are red, 
    while the inter-layer ones are blue, and 
    (b)~inter-bag edges; the intra-layer ones are green, 
    while the inter-layer ones are purple (forward) and orange (backward).}
    \label{fig:categories-edges}
\end{figure}

\begin{lemma}[Dujmović et al.~\cite{DBLP:journals/jacm/DujmovicJMMUW20}]
\label{lem:hpartitionqueus}
For all graphs $G$ and $H$, if $H$ admits a $k$-queue layout and $G$ has an $H$-partition 
of layered-width $\ell$ with respect to some layering $\mathcal{L}=(V_0,V_1,\ldots)$ of $G$, then $G$ admits a $(3\ell
k+\lfloor\frac{3}{2}\ell\rfloor)$-queue layout using vertex order
$\overrightarrow{V_0},\overrightarrow{V_1},\dots$, where $\overrightarrow{V_i}$ is some 
order of $V_i$. In particular,
\begin{equation}\label{eq:upper_bound}
\mathrm{qn}(G) \leq 3 \ell \,\mathrm{qn}(H) + \lfloor \tfrac{3}{2}\ell \rfloor.
\end{equation}
\end{lemma}
In the proof of \cref{lem:hpartitionqueus}, the order  of the vertices of $G$ on each layer of $\mathcal{L}$ is defined as follows. Let $x_1, \dots, x_h$ be the vertices of $H$ as they appear in a $k$-queue layout $\ql{H}$ of $H$ and let  $A_{x_1}, \dots, A_{x_h}$ be the corresponding bags of the $H$-partition. Then, the order $\overrightarrow{V_i}$ of each layer $V_i$ with $i\geq 0$ is:
\[
\overrightarrow{V_i}=A_{x_1}\cap V_i, A_{x_2}\cap V_i, \dots, A_{x_h}\cap V_i
\]
where each subset $A_{x_j}\cap V_i$ is ordered arbitrarily. This gives the total order $\lo{G}$ for the vertices of $G$. The edge-to-queue assignment, which completes the construction of the queue layout $\ql{G}$ of $G$,  exploits the following two properties: 
\begin{myenumerate}{(P.\arabic*)}
\item\label{p:intra-bag-nest} Two intra-bag edges nest in $\lo{G}$, only if they belong to the same bag; see blue and red edges in \cref{fig:intra_bag_edges}.
\item\label{p:inter-layer-nest} Two inter-layer edges nest in $\lo{G}$, only if their endpoints belong to the~same pair of consecutive layers of $\mathcal{L}$; see blue, purple and orange edges in \cref{fig:categories-edges}.
\end{myenumerate}
For the edge-to-queue assignment, the edges of $G$ are classified into four categories given by the bags of the $H$-partition and the layers of $\mathcal{L}$. We start with edges whose endpoints belong to the same bag (i.e., intra-bag edges); see \cref{fig:intra_bag_edges}.
\begin{myenumerate}{(E.\arabic*)}
\item \label{e:in_b_in_l} Intra-layer intra-bag edges of $G$ are assigned to at most $\lfloor \frac{\ell}{2} \rfloor$ queues, as the queue~number of $K_\ell$ is $\lfloor \frac{\ell}{2} \rfloor$~\cite{DBLP:conf/focs/Heath84}; see red edges in \cref{fig:intra_bag_edges}.
\item \label{e:in_b_out_l} Inter-layer intra-bag edges of $G$ are assigned to at most $\ell$ queues, as the queue~number of $K_{\ell, \ell}$  is $\ell$, when all vertices of the first bipartition precede those of the second; see blue edges in \cref{fig:intra_bag_edges}.
\end{myenumerate}

\noindent The remaining edges of $G$ connect vertices of different~bags (i.e., inter-bag edges); see \cref{fig:inter_bag_edges}. We further partition the inter-layer inter-bag edges into two categories. Let $(u,v)$ be an inter-layer inter-bag edge with $u \in A_x \cap V_i$ and $v \in A_y \cap V_{i+1}$, for some $i\geq0$. Then $(u,v)$ is \emph{forward}, if $x\; \lo{H}\; y$ holds in $\ql{H}$; otherwise, it is \emph{backward}; see purple and orange edges in \cref{fig:inter_bag_edges}, respectively.
For all inter-bag edges, in total, $3\ell k$ queues suffice (see~\cite[Lemma~9]{DBLP:journals/jacm/DujmovicJMMUW20} for~details).
\begin{myenumerate}{(E.\arabic*)}\setcounter{enumi}{2}
\item \label{e:out_b_in_l} Intra-layer inter-bag edges of $G$ are assigned to at most $\ell k$ queues; on each layer, an edge of $H$ corresponds to a subgraph of $K_{\ell, \ell}$, where the first bipartition precedes the second; see green edges in \cref{fig:inter_bag_edges}.
\item \label{e:out_b_out_l_for} Forward inter-layer inter-bag edges of $G$ are assigned to at most $\ell k$~queues; for two consecutive layers, an edge of $H$  corresponds to a subgraph of $K_{\ell, \ell}$, where the first bipartition precedes the second; see purple edges of \cref{fig:inter_bag_edges}.
\item \label{e:out_b_out_l_back} Symmetrically all backward inter-layer inter-bag edges of $G$ are assigned to at most $\ell k$ queues; see orange edges in \cref{fig:inter_bag_edges}.
\end{myenumerate}

\noindent The next property follows from the proof of Lemma 9 in \cite{DBLP:journals/jacm/DujmovicJMMUW20}:
\begin{myenumerate}{(P.\arabic*)}\setcounter{enumi}{2}
\item \label{p:inter-rainbow} 
For $1 \leq i \leq r$, let $(u_i,v_i)$ be an edge of $G$, such that $u_i\; \lo{G}\; v_i$, $u_i \in A_{x_i}$ and $v_i \in A_{y_i}$. If all these $r$ edges belong to one of \ref{e:out_b_in_l}-\ref{e:out_b_out_l_back} and form an $r$-rainbow in $\lo{G}$, while edges 
$(x_1,y_1),\ldots,(x_r,y_r)$ of $H$ are assigned to the same queue in $\ql{H}$, then $r \leq \ell$ and either $u_1,\ldots,\allowbreak u_r$ or $v_1,\ldots, v_r$ belong to the same bag of the $H$-partition of $G$.
\end{myenumerate}

\noindent If $G$ is maximal plane, few more ingredients are needed to apply \cref{lem:hpartitionqueus}.
A \emph{vertical path} of $G$ in a BFS-layering $\mathcal{L}$ is a path $P= v_0,\dots, v_k$ of $G$ consisting only of edges of the BFS-tree of $\mathcal{L}$ and such that if $v_0$ belongs to $V_i$ in $\mathcal{L}$, then $v_j$ belongs to $V_{i+j}$, with $j=1,\ldots,k$. Further, we say that $v_0$ and $v_k$ are the \emph{first} and \emph{last} vertices of $P$.
A \emph{tripod} of $G$ consists of up to three pairwise vertex-disjoint vertical paths in  $\mathcal{L}$ whose last vertices form a clique of size at most $3$ in $G$. We refer to this clique as the \emph{base} of the tripod.
 Dujmović et al.~\cite{DBLP:journals/jacm/DujmovicJMMUW20} showed that for any BFS-layering $\mathcal{L}$, $G$ admits an $H$-partition with the following properties:
\begin{myenumerate}{(P.\arabic*)}\setcounter{enumi}{3}
\item \label{prp:k} $H$ is a planar $3$-tree and thus $\ql{H}$ is a $k$-queue layout with $k \leq 5$~\cite{DBLP:journals/algorithmica/AlamBGKP20}.
\item \label{prp:l} Its layered-width $\ell$  is at most $3$, since each bag induces a \emph{tripod} in $G$, whose base is a triangular face of $G$, if it is a $3$-clique.
\end{myenumerate}
Properties~\ref{prp:k} and~\ref{prp:l} along with \cref{eq:upper_bound} imply that the queue number of planar graphs is at most $3 \cdot 3 \cdot 5 + \lfloor \tfrac{3}{2}\cdot3 \rfloor=49$.

\section{Refinements of the involved techniques}
\label{sec:all-refinements}
In this section, we present refinements of the algorithms outlined in \cref{sec:all-outlines} that will allow us to reduce the upper bound on the queue number of planar~graphs.

\subsection{Outerplanar Graphs}
\label{sec:outerplanar-refine}

We modify the algorithm by Dujmović et al.~\cite{DBLP:journals/dmtcs/DujmovicPW04} outlined in \Cref{sec:outerplanar-outline} to guarantee two additional properties (stated in \cref{lem:side-top}) of the outerplanar drawing. To this end, besides Invariants~\ref{outer:1} and~\ref{outer:2}, we maintain a third~one: 

\begin{myenumerate}{(O.\arabic*)}\setcounter{enumi}{2}
\item\label{outer:3}The lower envelope consists of a single edge. 
\end{myenumerate}

To maintain Invariant~\ref{outer:3}, we observe that a biconnected maximal outerplane graph with at least four vertices contains at least two non-adjacent degree-$2$ vertices. Let $x$ be such a degree-$2$ vertex of $G$, which we assume to be fixed in the recursive construction of $\Gamma(G)$. In particular, by our previous observation, we can always remove a degree-$2$ vertex that is different from $x$ at every recursive step. This ensures that $x$ will eventually appear in the triangle $\mathcal{T}$ that is drawn at the base of the recursion. We draw $\mathcal{T}$, such that $x$ is its bottommost vertex while its two incident edges are of span-1 and span-2 as follows. Let $y$ and $z$ be the other two vertices of $\mathcal{T}$. We draw $x$, $y$ and $z$ at $(2,0)$, $(0,1)$ and $(1,2)$, respectively. This implies that $(x,y)$ forms the lower envelope of $\mathcal{T}$, while $(y,z)$ and $(z,x)$ form the upper one. This guarantees that Invariant~\ref{outer:3} is maintained at the base of the recursion. Assume now that drawing $\Gamma(G')$ obtained by removing $v$ from $G$ satisfies Invariant~\ref{outer:3}, such that the edge $(x,y)$ of $\mathcal{T}$ forms the lower envelope of $\Gamma(G')$. Since $(x,y)$ belongs to the outer face of $G$, it follows that $v$ is incident to two vertices of the upper envelope of $\Gamma(G')$. So, after the addition of $v$ to $\Gamma(G')$ in order to obtain $\Gamma(G)$, edge $(x,y)$ forms the lower envelope of $\Gamma(G)$, as~desired.

Let $\langle u,v,w \rangle$ be a face of  $\Gamma(G)$ such that $y(u)-y(w)=2$ and $y(u)-y(v)=y(v)-y(w)=1$. We refer to vertices $u$, $v$ and $w$ as the \emph{top}, \emph{middle} and \emph{bottom} vertex of the face, respectively\footnote{Alam et al.~\cite{DBLP:journals/algorithmica/AlamBGKP20} refer to the middle vertex of a triangular face in $\Gamma(G)$ as its \emph{anchor}.}. Further, we say that face $\langle u,v,w \rangle$ is a \emph{bottom}, \emph{side} and \emph{top triangle} for vertices $u$, $v$ and $w$, respectively.

In the next lemma, we prove two properties of drawing $\Gamma(G)$. Part~\ref{i:top} of \cref{lem:side-top} requires Invariant~\ref{outer:3}, i.e., it does not necessarily hold for all drawings obtained by the algorithm by Dujmović et al.~\cite{DBLP:journals/dmtcs/DujmovicPW04}. On the other hand, part~\ref{i:side} of \cref{lem:side-top} holds for drawings that do not necessarily satisfy Invariant~\ref{outer:3}, i.e., it is a property of the original algorithm by Dujmović~et~al.~\cite{DBLP:journals/dmtcs/DujmovicPW04}. 

\begin{lemma}\label{lem:side-top}
Let $\Gamma(G)$ be an outerplanar drawing satisfying Invariants~\ref{outer:1}--\ref{outer:3} of a biconnected maximal outerplane graph $G$. Then, each vertex of $G$~is 
\begin{enumerate}[label=\normalfont(\alph*), align=left]
\item\label{i:top} the top vertex of at most two triangular faces of $\Gamma(G)$ and
\item\label{i:side} the side vertex of at most two triangular faces of $\Gamma(G)$. 
\end{enumerate}
\end{lemma}
\begin{proof}
For \ref{i:top}, consider a vertex $u$ of $G$. If $u$ is the top vertex of a face, then~$u$ is incident to a span-$2$ edge $(u,v)$ with $y(u)>y(v)$. By Invariant~\ref{outer:3}, $u$ is a successor of $v$ in the recursive approach by Dujmović et al.~\cite{DBLP:journals/dmtcs/DujmovicPW04}, i.e., when $u$ is placed in $\Gamma(G)$, vertex $v$ belongs to the upper envelope. We now claim that $u$ cannot be incident to two edges $(u,v)$ and $(u,v')$ with the properties mentioned above; this claim implies the lemma. Assuming the contrary, by Invariant~\ref{outer:2}, when $u$ is placed in $\Gamma(G)$ at most one edge incident to $u$ has span~$2$. So, at most one of $(u,v)$ and $(u,v')$ is drawn when $u$ is placed in $\Gamma(G)$, which implies that at least one of $v$ and $v'$, say $v'$, is a successor of $u$. Thus, $y(u)<y(v')$ holds; a contradiction. 

For \ref{i:side}, assume for a contradiction that $G$ contains a vertex $u$, which is the side vertex of  three triangular faces, say $\mathcal{T}_1$, $\mathcal{T}_2$ and $\mathcal{T}_3$, of $\Gamma(G)$. For a significantly small constant $\varepsilon>0$, let $p_1=(x(u)-\varepsilon,y(u))$ and $p_2=(x(u)+\varepsilon,y(u))$, and consider the two horizontal rays $r_1$ and $r_2$ emanating from vertex $u$, such that $r_1$ and $r_2$ contain points $p_1$ and $p_2$, respectively. Since $u$ is side vertex for $\mathcal{T}_1$, $\mathcal{T}_2$ and $\mathcal{T}_3$, it follows by Invariant~\ref{outer:2} that the span-$2$ edge of each of $\mathcal{T}_1$, $\mathcal{T}_2$ and $\mathcal{T}_3$ crosses either $r_1$ or $r_2$; by the choice of $\varepsilon$, we may assume that if a span-$2$ edge of one of $\mathcal{T}_1$, $\mathcal{T}_2$ and $\mathcal{T}_3$ crosses $r_1$ ($r_2$), then point $p_1$ ($p_2$, respectively) lies in the interior of it. On the other hand, by outerplanarity, each of the points $p_1$ and $p_2$ can be contained in at most one of $\mathcal{T}_1$, $\mathcal{T}_2$ and $\mathcal{T}_3$; a contradiction.
\end{proof}

\subsection{Planar 3-trees}
\label{sec:planar-3-trees-refine}

To maintain Invariant~\ref{t:level}, Alam et al.~\cite{DBLP:journals/algorithmica/AlamBGKP20} use the algorithm by Dujmović et al.~\cite{DBLP:journals/dmtcs/DujmovicPW04} to assign the level edges of  $L_0,\ldots,L_\lambda$ of the input plane $3$-tree~$H$ to two queues~$\Qh_0$ and $\Qh_1$, since on each level these edges induce a (not necessarily connected) outerplane~graph. 
Unlike in the original algorithm, in our approach we adopt the modification for the algorithm by Dujmović et al.~\cite{DBLP:journals/dmtcs/DujmovicPW04} introduced in \Cref{sec:outerplanar-refine}. As Invariants~\ref{outer:1} and~\ref{outer:2} are preserved, queues $\Qh_0$ and $\Qh_1$  suffice. 

To maintain Invariant~\ref{t:bind}, Alam et al.~\cite{DBLP:journals/algorithmica/AlamBGKP20} adopt the following assignment scheme for the binding edges between $L_j$ and $L_{j+1}$ to queues $\Qh_2$, $\Qh_3$ and $\Qh_4$, for each $j = 0, \ldots, \lambda-1$. Consider a binding edge $(u,v)$ with $u \in L_j$ and $v \in L_{j+1}$. Then, vertex $u$ belongs to a connected~component $C_u$ of the subgraph $H_j$ of $H$ induced by the level-$L_j$ vertices, while vertex $v$ belongs to a connected component $\mathcal C_v$ of $H_{j+1}$. Further, $C_u$ is outerplane and its $2$-queue layout has been computed by the algorithm by Dujmović et al.~\cite{DBLP:journals/dmtcs/DujmovicPW04}, while $C_v$ resides in the interior of a triangular face $\mathcal{T}_v$ of $C_u$ in the embedding of $H$, such that $u$ is on the boundary~of $\mathcal{T}_v$. 
In the edge-to-queue assignment scheme by Alam et al.~\cite{DBLP:journals/algorithmica/AlamBGKP20}, edge 
$(u,v)$ is assigned to $\Qh_2$, $\Qh_3$ or $\Qh_4$ if and only if $u$ is the top, middle or bottom vertex of $\mathcal{T}_v$, respectively~\cite{DBLP:journals/algorithmica/AlamBGKP20}. 
For a vertex $u\in L_j$ with $j=0,\ldots, \lambda-1$, we denote by $N_2(u)$, $N_3(u)$ and $N_4(u)$ the neighbors of $u$ in $L_{j+1}$ such that the edges connecting them to $u$ are assigned to queue $\Qh_2$, $\Qh_3$ and $\Qh_4$, respectively. The next property follows from \cref{lem:side-top}.

\begin{lemma}\label{lem:at-most-2}
For $j = 0, \ldots, \lambda-1$, let $u\in L_j$ be a vertex of $H$ in our modification of the peeling-into-levels approach by Alam et al.~\cite{DBLP:journals/algorithmica/AlamBGKP20}.
Then, the vertices of $N_2(u)$ precede those of $N_3(u)$, and the vertices of $N_2(u)$ (or of $N_3(u)$) belong to at most two connected components of $H_{j+1}$ (residing within distinct faces of $H_{j}$). 
\end{lemma}
\begin{proof}
The first part of the lemma is proven in~\cite{DBLP:journals/algorithmica/AlamBGKP20}.
For the second part, consider a binding edge $(u,v)$ with $v\in N_2(u)$; a similar argument applies when $v \in N_3(u)$. Thus, $u$ is the top vertex of the triangular face $\mathcal{T}_v$ of $H_j$, in which the connected component $C_v$ of $H_{j+1}$ that contains $v$ resides. Since by \cref{lem:side-top}\ref{i:top} vertex $u$ can be the top vertex of at most two triangular faces of $H_j$, there exist at most two connected components of $H_{j+1}$, to which vertex $v$ can belong.   
\end{proof}

\subsection{General Planar Graphs}
\label{sec:planar-refine}

Dujmović et al.~\cite{DBLP:journals/jacm/DujmovicJMMUW20} recursively compute the bags (i.e., the tripods) of the $H$-partition.
Each newly discovered tripod $\tau$ is adjacent to at most three other tripods $\tau_1$, $\tau_2$ and $\tau_3$  already discovered. We say that 
$\tau_1$, $\tau_2$ and $\tau_3$ are the \emph{parents} of $\tau$; see \cref{fig:tripod-1}. 
Also, each non-empty vertical path of $\tau$ is connected to only one of its parents via an edge of the BFS-tree used to construct the BFS-layering $\mathcal{L}$ (black in \cref{fig:tripod-1}). This property gives rise to at most three sub-instances (gray in \cref{fig:tripod-1}), which are processed recursively to compute the final tripod decomposition. 

Note that, in general, one or more vertical paths of a tripod may have no vertices; see \cref{fig:tripod-1-special}. However, by an appropriate augmentation of the input graph we can avoid considering this degenerate case explicitly. More precisely, we first compute a BFS tree $T$ of $G$ rooted at a vertex of the outer face. We next augment $G$ by adding a triangle $f^*$ in the interior of each face $f$ of $G$. By an appropriate triangulation, we can obtain a triangulated plane graph $G^*$ that contains $G$ as a subgraph, such that each of the newly introduced vertices of $f^*$ can be added as a leaf of $T$ to derive a BFS tree $T^*$ of $G^*$.  In this way, each original face $f$ of $G$ gets three leaves of $T^*$ inside it, one attached to each vertex of $f$. By computing a tripod decomposition of $G$ using the tree $T$, one obtains a collection of tripods. Each of these tripods has a ``base'', which is a triangular face $f$ of $G$.  We produce a tripod decomposition of $G^*$ with respect to $T^*$ by extending each tripod with base $f$ so that its base is $f^*$. In this way, we derive a tripod decomposition of $G^*$, in which every tripod has three non-empty vertical paths, as desired. Hence, without loss of generality, we assume in the following that in the tripod decomposition no vertical path of a tripod is empty.

For $i\in \{1,2,3\}$, let $p_i^1$, $p_i^2$ and $p_i^3$ be the three vertical paths of tripod $\tau_i$. Up to renaming, we assume that $\tau$ lies in the cycle bounded by (parts of) $p_1^1$, $p_1^2$, $p_2^1$, $p_2^2$, $p_3^1$ and $p_3^2$ as in \cref{fig:tripod-1}. 
The next properties follow by planarity and the BFS-layering:

\begin{figure}[t]
   \centering
   \begin{subfigure}{0.35\textwidth}
        \centering
        \includegraphics[page=2]{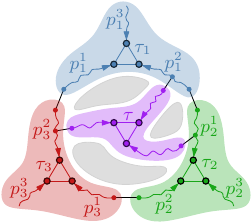}
        \caption{$\tau$ shares no vertex with $\tau_1$, $\tau_2$, $\tau_3$ }
    \label{fig:tripod-1}
    \end{subfigure}
    \hfil
    \begin{subfigure}{0.32\textwidth}
        \centering
       \includegraphics[page=11]{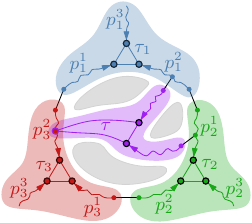}
        \caption{$\tau$ shares a vertex with $\tau_3$}
        \label{fig:tripod-1-special}
    \end{subfigure}
    \hfil
    \begin{subfigure}{0.3\textwidth}
        \centering
       \includegraphics[page=3]{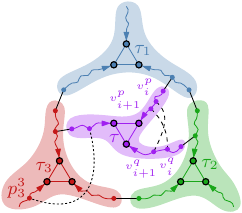}
        \caption{}
        \label{fig:tripod-2}
    \end{subfigure}
    \caption{
    (a), (b)~Tripod $\tau$ with parents $\tau_1$, $\tau_2$ and $\tau_3$.
    (c)~Illustration for~\ref{prp:t_parent} and~\ref{prp:t_xing_edges}.
    }
    \label{fig:tripods}
\end{figure}

\begin{myenumerate}{(P.\arabic*)}\setcounter{enumi}{5}
\item \label{prp:t_parent} There is no edge connecting a vertex of $\tau$ to a vertex of $p_i^3$ for $i=1,2,3$; see the dotted edge in~\cref{fig:tripod-2}.
\item \label{prp:t_bfs} Let $v_i^p$ be the vertex of vertical path $p$ of $\tau$ on layer $V_i$ of $\mathcal{L}$. For two vertical paths $p$ and $q$ of $\tau$, edge $(v_i^p,v_j^{q})$ belongs to $G$ only if $|i-j|\leq 1$.
\item \label{prp:t_xing_edges} For vertical paths $p$ and $q$ of $\tau$, at most one of the edges $(v_i^p,v_{i+1}^q)$ and $(v_{i+1}^p,v_{i}^q)$ exists in $G$; see the dashed edges of~\cref{fig:tripod-2}.

\end{myenumerate}

\noindent Note that \ref{prp:t_parent}--\ref{prp:t_xing_edges} hold even if $\tau$ has less than three parents, or if the cycle bounding the region of $\tau$ does not contain two vertical paths of each parent~tripod.

In the original algorithm by Dujmović et al.~\cite{DBLP:journals/jacm/DujmovicJMMUW20}, each vertex $v_\tau$ in $H$ corresponds to a tripod $\tau$ in $G$, and an edge $(v_\tau,v_{\tau'})$ exists in $H$, if $\tau$ is a parent of $\tau'$ in $G$, or vice versa. 
Also, $H$ is a connected partial planar $3$-tree, which is arbitrarily~augmented to a maximal planar $3$-tree $H'$ (to compute a 5-queue layout of it).
Here, we adopt a particular augmentation to guarantee an additional property for the graph $H'$ (see \Cref{lem:parent-child}). 
Similarly to the original approach, we contract the vertices of each tripod $\tau$ of $G$ to a single vertex $v_\tau$. However, in our modification, we keep self-loops that occur when an edge of $G$ has both endpoints in $\tau$ (unless this edge belongs to one of its vertical paths), as well as, parallel edges that occur when two vertices of $\tau$ have a common neighbor not~in~$\tau$. Two important properties of this contraction scheme that follow directly from planarity are given below; see \cref{fig:tripods-contract}.

\begin{myenumerate2}{(P.\arabic*)}\setcounter{enumi}{8}
\item \label{prp:t_cyclic_order_all}
The edges around each contracted vertex $v_\tau$ appear in the same clockwise cyclic order as they appear in a clockwise traversal along $\tau$ in $G$. 
\item \label{prp:t_cyclic_order_path} 
The edges having at least one endpoint on the same vertical path of $\tau$ appear consecutively around $v_\tau$; see \cref{fig:tripod-contract-2}.
\end{myenumerate2}

\begin{figure}[t]
   \centering
   \begin{subfigure}{0.4\textwidth}
        \centering
        \includegraphics[page=7, width=\textwidth]{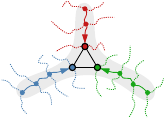}
        \caption{}
    \label{fig:tripod-contract-1}
    \end{subfigure}
    \hfil
    \begin{subfigure}{0.4\textwidth}
        \centering
       \includegraphics[page=8, width=\textwidth]{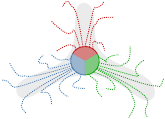}
        \caption{}
        \label{fig:tripod-contract-2}
    \end{subfigure}
    \caption{
    (a)~A tripod $\tau$ in $G$ where the edges incident to its three vertical paths are drawn dotted, and
    (b)~the result after contracting $\tau$ to $v_\tau$.
    }
    \label{fig:tripods-contract}
\end{figure}

\noindent To guarantee simplicity, 
we focus on \emph{homotopic} self-loops and pairs of parallel edges, which contain no vertex either in the interior or in the exterior region that they define. We remove such self-loops and keep one copy of such parallel edges. 
Then, we subdivide each self-loop twice, and for each edge with multiplicity $m>1$, we subdivide all but one of its copies. 
In this way, each vertex $v_\tau$ corresponding to a tripod $\tau$ in $G$ always lies in the interior of a separating $3$-cycle $C$ that contains all the vertices corresponding to its parent tripods on its boundary. 
To see this, observe that if $\tau$ has three parent $\tau_1$, $\tau_2$ and $\tau_3$, then $C$ is formed by $v_{\tau_1}$, $v_{\tau_2}$ and $v_{\tau_3}$, such that $v_\tau$ is connected to each of them. If $\tau$ has two parent $\tau_1$ and $\tau_2$, then $C$ is formed by $v_{\tau_1}$, $v_{\tau_2}$ and a subdivision vertex, such that $v_\tau$ is connected to $v_{\tau_1}$ and $v_{\tau_2}$. Finally, if $\tau$ has only one parent $\tau_1$, then $C$ is formed by $v_{\tau_1}$ and two subdivision vertices, such that $v_\tau$ is connected to~$v_{\tau_1}$.
Since subdivision vertices are of degree $2$, the result is a simple (possibly not maximal) planar 3-tree, which is a supergraph of~$H$. 
To derive $H'$, we augment it to maximal by adding edges, while maintaining its~embedding~\cite{DBLP:journals/corr/abs-1210-8113}.

\begin{lemma}\label{lem:parent-child}
Let $v_\tau$ and $v_{\tau_p}$ be two vertices of $H'$ that correspond to a tripod $\tau$~and to a parent tripod $\tau_p$ of $\tau$ in $G$. 
If $L_i$ and $L_j$ are the levels of $v_\tau$ and $v_{\tau_p}$ in the peeling-into-levels approach for~$H'$, then $i\geq j$.
\end{lemma}
\begin{proof}
Let $C$ be the inclusion-minimal separating $3$-cycle of $H'$ containing $v_\tau$ in its interior and all  vertices that correspond to the parent tripods of $\tau$ on its boundary.  Let $L_k$, $L_l$ and $L_m$ be the levels of the three vertices of $C$, with $k \leq l \leq m$, in the peeling-into-levels approach for~$H'$. As $\tau_p$ is a parent~of~$\tau$, $j \in \{k,l,m\}$~holds. Since $C$ is a $3$-cycle and since each edge in the peeling-into-level approach is~either level or binding, $m\leq k+1$ holds. 
The fact that $v_\tau$ lies in the interior of $C$ and~is connected to each of the vertices in $H'$ corresponding to its parent tripods in~$G$, implies that $v_\tau$ is on level $L_{k+1}$, i.e., $i=k+1$. So, $j \leq m \leq k+1=i$ holds.
\end{proof}

As in the original algorithm by Dujmović et al.~\cite{DBLP:journals/jacm/DujmovicJMMUW20}, we compute a $5$-queue layout $\ql{H'}$ of $H'$. However, we use our modification of the algorithm by Alam et al.~\cite{DBLP:journals/algorithmica/AlamBGKP20} described in \cref{sec:planar-3-trees-refine}. Denote by $x_1, \dots, x_h$ the vertices of the subgraph $H$ of $H'$ as they appear in $\ql{H'}$ (i.e., we ignore subdivision vertices introduced when augmenting $H$ to $H'$) and by $\Qh_0,\ldots,\Qh_4$ the queues of $\ql{H'}$ as described in Invariants~\ref{t:order}--\ref{t:bind}. To compute the linear layout $\ql{G}$ of $G$, we use \cref{lem:hpartitionqueus}, which orders the vertices of each layer $V_i$, with $i \geq 0$, of $\mathcal{L}$ as:
\[
\overrightarrow{V_i}=A_{x_1}\cap V_i, A_{x_2}\cap V_i, \dots, A_{x_h}\cap V_i,
\]
where $A_{x_1},\ldots,A_{x_h}$ are the bags (i.e., the tripods) of the $H$-partition of $G$. 
Unlike in the original algorithm, we do not order the vertices in each subset $A_{x_j}\cap V_i$, with $j\in\{1,\ldots,h\}$, arbitrarily. Instead, we carefully choose their order; we describe this choice in the rest of this section.
Let $\tau$ be the tripod of bag $A_{x_j}$. Then, $A_{x_j}\cap V_i$ contains at most one vertex of each vertical path of $\tau$. We will order the three vertical paths of $\tau$, which defines the order of the (at most three) vertices of $A_{x_j}\cap V_i$  for every $i\geq 0$. 

Let $L_l$ be the level of $v_\tau$ in the peeling-into-levels of $H'$, with $0 \leq l< \lambda $. 
By \cref{lem:at-most-2}, there are at most four connected components~$c_s^1$, $c_s^2$, $c_t^1$ and $c_t^2$ of the subgraph $H_{l+1}'$ of $H'$ induced by the vertices of $L_{l+1}$, such that the edges connecting $v_\tau$ to vertices of $c_s^1$ and $c_s^2$ ($c_t^1$ and $c_t^2$) belong to $\Qh_2$ ($\Qh_3$, respectively).
Let $c$ be one of $c_s^1$, $c_s^2$, $c_t^1$ and $c_t^2$;  $c$ may contain vertices that correspond~to~tripods in $G$ (i.e., not to subdivisions introduced while augmenting $H$ to $H'$). We refer to the union of these vertices of $G$ as the \emph{tripod-vertices} of $c$.
By Invariant~\ref{t:order},~$v_\tau$ precedes the vertices of $c_s^1$, $c_s^2$, $c_t^1$ and $c_t^2$ in $\ql{H'}$. Also by \ref{t:order_component}, we~may assume that the vertices of $c_s^1$ ($c_t^1$) precede those of $c_s^2$ ($c_t^2$,~respectively). Additionally, \cref{lem:at-most-2} ensures that the vertices of $c_s^2$ precede those~of~$c_t^1$.

Since an edge $(v_\tau,v_{\tau'})$ exists in $H$, if $\tau$ is a parent of $\tau'$, or vice versa, by \cref{lem:parent-child}, for each vertex $v_{\tau'}$ of $H$ that is a neighbor of $v_\tau$ in one of $c_s^1$, $c_s^2$, $c_t^1$ and $c_t^2$, it follows that $\tau$ is a parent of $\tau'$.
By Property~\ref{prp:t_parent}, there is a vertical path of $\tau$ in $G$, say $p$ ($q$), such that no tripod-vertex of $c_s^2$ ($c_t^2$, respectively) is adjacent to it in $G$. Note that $p$ and $q$ might be the same vertical path of $\tau$.

We now describe the order of the three vertical paths of $\tau$. We only specify the first one; the other two can be arbitrarily ordered:
\begin{enumerate*}[label={(\roman*)}]
\item\label{first:1} if the tripod-vertices of $c_s^1$ and $c_s^2$ are connected to all three vertical paths of $\tau$ in $G$, then  $p$ is the first vertical path of $\tau$;
\item\label{first:2} if the tripod-vertices of $c_t^1$ and $c_t^2$ are connected to all three vertical paths of $\tau$ in $G$, then $q$ is the first vertical path of $\tau$;
\item\label{first:3} otherwise, any vertical path of $\tau$ can be first.
\end{enumerate*}
Before proving formally that Cases~\ref{first:1} and~\ref{first:2} cannot apply simultaneously (which implies that the order of the vertical paths of $\tau$ is well defined), we state two important implications of the described choice for the first vertical path of $\tau$. 
\begin{myenumerate2}{(P.\arabic*)}\setcounter{enumi}{10}
\item \label{prp:q2} Under our assumption that all vertices of $c_s^1$ precede those of $c_s^2$, if tripod-vertices of $c_s^1$ and $c_s^2$ are connected to all three vertical paths of $\tau$, then tripod-vertices of $c_s^2$ are not connected to the first vertical path of $\tau$. 
\item \label{prp:q3} Also, under our assumption that all vertices of $c_t^1$ precede those of $c_t^2$, if tripod-vertices of $c_t^1$ and $c_t^2$ are connected to all three vertical paths of $\tau$, then tripod-vertices of $c_t^2$ are not connected to the first vertical path of $\tau$.
\end{myenumerate2}

\begin{figure}[t]
   \centering
   \begin{subfigure}[b]{0.32\textwidth}
        \centering
        \includegraphics[page=9]{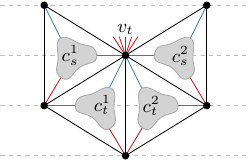}
        \caption{}
    \label{fig:two-bottom-components}
    \end{subfigure}
    \hfil
    \begin{subfigure}[b]{0.32\textwidth}
        \centering
        \includegraphics[page=10]{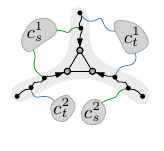}
        \caption{}
        \label{fig:alternating-edges}
    \end{subfigure}
    \begin{subfigure}[b]{0.32\textwidth}
        \centering
        \includegraphics[page=13]{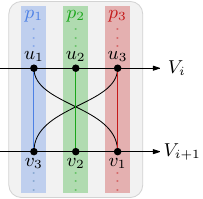}
        \caption{}
        \label{fig:reduce_1}
   \end{subfigure}
    \caption{
    Illustrations for: (a)-(b)~the fact that Cases~\ref{first:1} and~\ref{first:2} in the selection of the first vertical path of $\tau$ cannot apply simultaneously, and
    (c)~the proof of \Cref{lem:reduce_1}.
    }
    \label{fig:tripods-edges-app}
\end{figure}

We now prove that Cases~\ref{first:1} and~\ref{first:2} in the selection of the first vertical path of $\tau$ cannot apply simultaneously.
Assume for a contradiction that the tripod-vertices of $c_s^1$ and $c_s^2$ are connected to all three vertical paths of $\tau$, and also that the tripod-vertices of $c_t^1$ and $c_t^2$ are connected to all three vertical paths of $\tau$.
Since $v_\tau$ is the top vertex of the two triangular faces $f_t^1$ and $f_t^2$ that contain $c_t^1$ and $c_t^2$ in their interior in $\Gamma(H'_l)$, it follows that the edges connecting $v_\tau$ to  $c_t^1$ and $c_t^2$ appear consecutively around $v_\tau$ in the cyclic order of the binding edges between levels $L_l$ and $L_{l+1}$ of $H'$; refer to the blue edges of \cref{fig:two-bottom-components}. 
Also, since $v_\tau$ is the side vertex of the two triangular faces $f_s^1$ and $f_s^2$ that contain $c_s^1$ and $c_s^2$ in their interior in $\Gamma(H'_l)$, we may assume w.l.o.g.\ that the edges connecting $v_\tau$ to  $c_s^1$ immediately precede the edges connecting $v_\tau$ to  $c_t^1$ and $c_t^2$, while the edges connecting $v_\tau$ to  $c_s^2$ immediately follow them; refer to the green edges of \cref{fig:two-bottom-components} (note that since we have assumed that the vertices of $c_s^1$ precede the vertices of $c_s^2$, our assumption on the order of the edges around $v_\tau$ is a property of the algorithm by Alam et al.~\cite{DBLP:journals/algorithmica/AlamBGKP20}). 
Now, if the tripod-vertices of $c_s^1$ and $c_s^2$ are connected to all three vertical paths of $\tau$, and the same holds for the tripod-vertices of $c_t^1$ and $c_t^2$, then Property~\ref{prp:t_cyclic_order_path} implies that in $G$ the edges connecting vertices of $\tau$ to tripod-vertices of $c_t^1$ and $c_t^2$ alternate with edges connecting vertices of $\tau$ to tripod-vertices of $c_s^1$ and $c_s^2$; see e.g. \cref{fig:alternating-edges}. Hence, the cyclic order of the edges around $v_\tau$ in $H'$ is not the same as the cyclic order of the edges incident to vertices of $\tau$ along a clockwise traversal of $\tau$ in $G$, contradicting Property~\ref{prp:t_cyclic_order_all}. Hence, Case~\ref{first:1} and~\ref{first:2} cannot apply simultaneously, as we initially claimed.

\section{Reducing the bound}
\label{sec:main}

To reduce the bound, we turn our attention to the analysis of the required number of queues for the intra-bag inter-layer edges~\ref{e:in_b_out_l}, as well as for the inter-bag edges (i.e., either intra-layer~\ref{e:out_b_in_l} or inter-layer~\ref{e:out_b_out_l_for}--\ref{e:out_b_out_l_back}) given by Dujmović et al.~\cite{DBLP:journals/jacm/DujmovicJMMUW20}; refer to \cref{sec:planar-outline}.
For intra-bag inter-layer edges~\ref{e:in_b_out_l}, the original algorithm by Dujmović et al.~\cite{DBLP:journals/jacm/DujmovicJMMUW20} uses three queues, since $\ell=3$; see blue edges in~\cref{fig:intra_bag_edges}. We prove that no three intra-bag inter-layer edges form a $3$-rainbow, implying that the upper bound on the queue~number of planar graphs can be improved from $49$~to~$48$. 

\begin{lemma}\label{lem:reduce_1}
In the queue layout computed by our modification of the algorithm by Dujmović et al.~\cite{DBLP:journals/jacm/DujmovicJMMUW20}, no three intra-bag inter-layer edges of $G$ form a $3$-rainbow. 
\end{lemma}
\begin{proof}
Assume to the contrary that there exist three such edges $(u_1,v_1)$, $(u_2,v_2)$ and $(u_3,v_3)$ forming a $3$-rainbow in $\ql{G}$ so that $u_1 \prec_G u_2 \prec_G u_3 \prec_G v_3 \prec_G v_2 \prec_G v_1$.
By \ref{p:intra-bag-nest} these edges belong to the same bag $A$ of the $H$-partition, while by \ref{p:inter-layer-nest}  their endpoints belong to two consecutive layers $V_i$ and $V_{i+1}$ of $\mathcal{L}$.
Due to the order, $u_1,u_2,u_3 \in V_i$ and $v_1,v_2,v_3 \in V_{i+1}$. 
The order of $A\cap V_i$~and~$A\cap V_{i+1}$~is  $u_1 \prec_G u_2 \prec_G u_3$ and $v_3 \prec_G v_2 \prec_G v_1$; see \cref{fig:reduce_1}.
Let $p_1$, $p_2$ and $p_3$ be the first, second and third vertical paths of tripod $\tau$ forming $A$. Then, $(u_1,v_3)\in p_1$, $(u_2,v_2)\in p_2$ and $(u_3,v_1)\in p_3$. 
However, $(u_1,v_1)$  and $(u_3,v_3)$ contradict~\ref{prp:t_xing_edges}.
\end{proof}

\noindent For inter-bag edges \ref{e:out_b_in_l}--\ref{e:out_b_out_l_back}, the algorithm by Dujmović et al.~\cite{DBLP:journals/jacm/DujmovicJMMUW20} uses $3\cdot 15$ queues, since $k=5$ and $\ell=3$; see green, purple and orange edges in~\cref{fig:inter_bag_edges}. We exploit \ref{p:inter-rainbow} to prove that $3\cdot 13$~queues suffice. This further improves the upper bound on the queue~number of planar graphs from~$48$~to~$\bound$. 

\begin{lemma}\label{lem:reduce_2}
In the queue layout computed by our modification of the algorithm by Dujmović et al.~\cite{DBLP:journals/jacm/DujmovicJMMUW20}, the inter-bag edges of $G$ do not form a $40$-rainbow.
\end{lemma}
\begin{proof}
To prove the statement, it suffices to show that the inter-bag edges of $G$ that belong to each of \ref{e:out_b_in_l}--\ref{e:out_b_out_l_back} form at most a $13$-rainbow.
We focus on the edges of \ref{e:out_b_in_l}, that is, on the intra-layer inter-bag edges of $G$; the other two types of edges, i.e., forward and backward inter-layer inter-bag edges, can be treated similarly. 
We partition the edges of \ref{e:out_b_in_l} into five subsets $E_3^0,\ldots,E_3^4$, such that for $i=0,\ldots,4$, set $E_3^i$ contains the following subset of \ref{e:out_b_in_l}-edges~of~$G$: 
\[
E_3^{i} =\{(u,v) \in \mbox{\ref{e:out_b_in_l}}: u\in A_x, v\in A_y, (x,y) \in \Qh_i  \}.
\]
By Property~\ref{p:inter-rainbow}, the edges of $E_3^i$ can form at most a $3$-rainbow (this property actually implies the initial bound of $15$ queues for all intra-layer inter-bag edges of $G$ by Dujmović et al.~\cite{DBLP:journals/jacm/DujmovicJMMUW20}). We next prove that neither the edges of $E_3^2$ nor the edges of $E_3^3$ form a $3$-rainbow, which yields the desired reduction from $15$ to $13$ on the size of the maximum rainbow formed by the edges of \ref{e:out_b_in_l}. 
Assume that this is not true, say for the edges of $E_3^2$; a symmetric argument applies for the edges of $E_3^3$. 
Let $(u_1,v_1)$, $(u_2,v_2)$ and $(u_3,v_3)$ be three edge of $E_3^2$ that form a $3$-rainbow, such that $u_1\prec_G u_2 \prec_G u_3 \prec_G v_3\prec_G v_2\prec_G v_1$. 
Since $(u_1,v_1)$, $(u_2,v_2)$ and $(u_3,v_3)$ are \ref{e:out_b_in_l}-edges, i.e., intra-layer edges, their endpoints $u_1$, $u_2$, $u_3$, $v_1$, $v_2$ and $v_3$ all belong to the same layer of the BFS-layering of $G$\footnote{Note that if we had assumed that $(u_1,v_1)$, $(u_2,v_2)$ and $(u_3,v_3)$ belonged to one of \ref{e:out_b_out_l_for} or \ref{e:out_b_out_l_back}, i.e., inter-layer edges, vertices $u_1$, $u_2$ and $u_3$ would all belong to one layer, while vertices $v_1$, $v_2$ and $v_3$ would belong to the next layer. In what follows, we will only use the fact that vertices $u_1$, $u_2$ and $u_3$ belong to the same layer and that vertices $v_1$, $v_2$ and $v_3$ belong to the same layer, without discriminating whether the two layers are the same or not.}.
Assuming that $u_i \in A_{x_i}$ and $v_i \in A_{y_i}$, it follows that $(x_i,y_i)$ is an edge of $H$ assigned to $\Qh_2$ and $x_i\prec_H y_i$. 
By Property~\ref{p:inter-rainbow}, either $x_1=x_2=x_3$ or $y_1=y_2=y_3$ holds.
First, we prove that $x_1=x_2=x_3$ always holds.
Assume for a contradiction that, w.l.o.g., $x_1\neq x_2$.
Then, $y=y_1=y_2=y_3$ holds.
Since the edges of $\Qh_2$ are binding in the peeling-into-levels decomposition of $H'$ and since $x_1\prec_H x_2\prec_H y$, vertices $x_1$ and $x_2$ belong to level $L_j$ of $H'$, while $y$ belongs to level $L_{j+1}$, for some $0\leq j<\lambda$. 
By Property~\ref{t:bind}, $y$ has only one neighbor on level $L_j$ such that the edge connecting $y$ to this neighbor is assigned to $\Qh_2$. This contradicts the fact that both edges $(x_1,y)$ and $(x_2,y)$ are assigned to $\Qh_2$, since $x_1 \neq x_2$. Hence $x=x_1=x_2=x_3$ holds. 

Now, since the edges of $\Qh_2$ are binding in the peeling-into-levels decomposition of $H'$, and since, for $i=1,2,3$ , $x \prec_H y_i$, it follows that vertex $x$ belongs to level $L_j$ of $H'$, while vertex $y_i$ belongs to level $L_{j+1}$, for some $0\leq j<\lambda$. 
Let $\tau$ be the tripod of bag $A_x$ and let $\tau_1$, $\tau_2$, $\tau_3$ be the tripods of $A_{y_1}$, $A_{y_2}$ and $A_{y_3}$, respectively. Note that vertices $y_1$, $y_2$ and $y_3$ are not necessarily distinct, that is, the corresponding tripods $\tau_1$, $\tau_2$ and $\tau_3$ are not necessarily pairwise different. However, since $u_1\prec_G u_2\prec_G u_3$, and since  $u_1$, $u_2$ and $u_3$ belong to the same layer in the BFS-layering $\mathcal{L}$ of $G$, we can conclude that $u_1$, $u_2$ and $u_3$ belong to $p_1$, $p_2$ and $p_3$, respectively, where $p_1$, $p_2$ and $p_3$ are the first, second and third vertical paths of $\tau$.
Since, for $i=1,2,3$, the edge $(x,y_i)$ belongs to  $H$, it follows that $\tau$  is a  parent tripod of $\tau_i$ or vice versa. Since $x$ belongs to level $L_j$ of $H'$ and vertices $y_1$, $y_2$ and $y_3$ belong to level $L_{j+1}$, by \cref{lem:parent-child} we conclude that $\tau$ is a parent tripod of $\tau_1$, $\tau_2$ and $\tau_3$.
By \cref{lem:at-most-2},  there exist at most two connected components $c_s^1$ and $c_s^2$ of $H'_{j+1}$, such that the edges connecting $x$ to vertices of $c_s^1$ and $c_s^2$ have been assigned to $\Qh_2$. W.l.o.g, by Invariant~\ref{t:order_component} we assume  that all vertices of $c_s^1$  (if any) precede those of $c_s^2$ (if any).

We first argue that not all vertices $y_1$, $y_2$ and $y_3$ belong to only one of $c_s^1$ or $c_s^2$, which further implies that both components $c_s^1$ and $c_s^2$ exist. Assume to the contrary that $y_1$, $y_2$ and $y_3$ belong to $c_s^1$; a similar argument applies for $c_s^2$.
By Property~\ref{prp:t_parent}, tripod-vertices of component $c_s^1$ are not adjacent to the vertices of at least one vertical path of $\tau$ in $G$, say $p_3$. In this case, however, we obtain a contradiction to the fact that $v_3$, which is a tripod-vertex of $c_s^1$, is connect to $u_3$ that belongs to $p_3$, i.e., edge $(u_3,v_3)$ cannot exist in $G$. 

Hence, not all vertices $y_1$, $y_2$ and $y_3$ belong to only one of $c_s^1$ or $c_s^2$, as we initially claimed. In particular, since $y_3 \preceq_H y_2 \preceq_H y_1$ and since all vertices of $c_s^1$ precede those of $c_s^2$, it follows that $y_3$ belongs to $c_s^1$, while $y_1$ belongs to $c_s^2$. Since $u_1$, $u_2$ and $u_3$ belong to $p_1$, $p_2$ and $p_3$, respectively, the vertices of $c_s^1$ and $c_s^2$ are connected to all vertical paths of $\tau$. Thus, Property~\ref{prp:q2} applies, which implies that the tripod-vertices of $c_s^2$ (including $v_1$) are not connected to the first vertical path $p_1$ of $\tau$, which is a contradiction to the fact that $(u_1,v_1)$ is an edge of $G$, since $u_1\in p_1$. Hence, the edges of $E_3^2$ cannot form a $3$-rainbow. Symmetrically, we can prove that the edges of $E_3^3$ cannot form a $3$-rainbow, which completes the proof of the lemma.
\end{proof}

\medskip\noindent We are now ready to state the main theorem of this section.

\begin{theorem}\label{thm:main}
Every planar graph has queue number at most $\bound$.
\end{theorem}

\section{Conclusions}
\label{sec:conclusions}

In this work, we improved the upper bound on the queue number of planar graphs by refining the three techniques involved in the original~algorithm~\cite{DBLP:journals/jacm/DujmovicJMMUW20}. 
We believe that our approach has the potential to further reduce the~upper bound by at least $3$  (i.e., from $\bound$ to $39$). However, more elaborate~arguments that exploit deeper the planarity of the graph are required, and several details need to be worked out.
Still the gap with the lower bound of $4$ remains large and needs to be further reduced. In this regard, determining the exact queue number of planar 3-trees becomes critical, since an improvement of the current upper bound of $5$ (to meet the lower bound of $4$) directly implies a corresponding improvement on the upper bound on the queue number of general planar graphs. 
On the other hand, to obtain a better understanding of the general open problem, it is also reasonable to further examine subclasses of planar graphs, such as bipartite planar graphs or planar graphs with bounded degree (e.g., max-degree~3).

\bibliographystyle{splncs03}
\bibliography{general,stacks,queues}

\begin{thebibliography}{10}
\providecommand{\url}[1]{\texttt{#1}}
\providecommand{\urlprefix}{URL }

\bibitem{DBLP:journals/algorithmica/AlamBGKP20}
Alam, J.M., Bekos, M.A., Gronemann, M., Kaufmann, M., Pupyrev, S.: Queue
  layouts of planar 3-trees. Algorithmica  82(9),  2564--2585 (2020),
  \url{https://doi.org/10.1007/s00453-020-00697-4}

\bibitem{DBLP:journals/algorithmica/BannisterDDEW19}
Bannister, M.J., Devanny, W.E., Dujmovic, V., Eppstein, D., Wood, D.R.: Track
  layouts, layered path decompositions, and leveled planarity. Algorithmica
  81(4),  1561--1583 (2019), \url{https://doi.org/10.1007/s00453-018-0487-5}

\bibitem{DBLP:journals/ita/BarthPRR95}
Barth, D., Pellegrini, F., Raspaud, A., Roman, J.: On bandwidth, cutwidth, and
  quotient graphs. {RAIRO} Theor. Informatics Appl.  29(6),  487--508 (1995),
  \url{https://doi.org/10.1051/ita/1995290604871}

\bibitem{DBLP:journals/siamcomp/BekosFGMMRU19}
Bekos, M.A., F{\"{o}}rster, H., Gronemann, M., Mchedlidze, T., Montecchiani,
  F., Raftopoulou, C.N., Ueckerdt, T.: Planar graphs of bounded degree have
  bounded queue number. {SIAM} J. Comput.  48(5),  1487--1502 (2019),
  \url{https://doi.org/10.1137/19M125340X}

\bibitem{DBLP:journals/jocg/KaufmannBKPRU20}
Bekos, M.A., Kaufmann, M., Klute, F., Pupyrev, S., Raftopoulou, C.N., Ueckerdt,
  T.: Four pages are indeed necessary for planar graphs. J. Comput. Geom.
  11(1),  332--353 (2020),
  \url{https://journals.carleton.ca/jocg/index.php/jocg/article/view/504}

\bibitem{DBLP:journals/jct/BernhartK79}
Bernhart, F., Kainen, P.C.: The book thickness of a graph. J. Comb. Theory,
  Ser. {B}  27(3),  320--331 (1979),
  \url{https://doi.org/10.1016/0095-8956(79)90021-2}

\bibitem{DBLP:journals/jgt/ChinnCDG82}
Chinn, P.Z., Chvatalova, J., Dewdney, A.K., Gibbs, N.E.: The bandwidth problem
  for graphs and matrices - a survey. J. Graph Theory  6(3),  223--254 (1982),
  \url{https://doi.org/10.1002/jgt.3190060302}

\bibitem{DBLP:journals/siamcomp/BattistaFP13}
{Di Battista}, G., Frati, F., Pach, J.: On the queue number of planar graphs.
  {SIAM} J. Comput.  42(6),  2243--2285 (2013),
  \url{https://doi.org/10.1137/130908051}

\bibitem{DBLP:journals/csur/DiazPS02}
D{\'{\i}}az, J., Petit, J., Serna, M.J.: A survey of graph layout problems.
  {ACM} Comput. Surv.  34(3),  313--356 (2002),
  \url{https://doi.org/10.1145/568522.568523}

\bibitem{DBLP:journals/jgaa/DujmovicF18}
Dujmovi{\'{c}}, V., Frati, F.: Stack and queue layouts via layered separators.
  J. Graph Algorithms Appl.  22(1),  89--99 (2018),
  \url{https://doi.org/10.7155/jgaa.00454}

\bibitem{DBLP:conf/focs/DujmovicJMMUW19}
Dujmovi{\'{c}}, V., Joret, G., Micek, P., Morin, P., Ueckerdt, T., Wood, D.R.:
  Planar graphs have bounded queue-number. In: Zuckerman, D. (ed.) {FOCS}. pp.
  862--875. {IEEE} Computer Society (2019),
  \url{https://doi.org/10.1109/FOCS.2019.00056}

\bibitem{DBLP:journals/jacm/DujmovicJMMUW20}
Dujmovic, V., Joret, G., Micek, P., Morin, P., Ueckerdt, T., Wood, D.R.: Planar
  graphs have bounded queue-number. J. {ACM}  67(4),  22:1--22:38 (2020),
  \url{https://dl.acm.org/doi/10.1145/3385731}

\bibitem{DBLP:journals/siamcomp/DujmovicMW05}
Dujmovi{\'{c}}, V., Morin, P., Wood, D.R.: Layout of graphs with bounded
  tree-width. {SIAM} J. Comput.  34(3),  553--579 (2005),
  \url{https://doi.org/10.1137/S0097539702416141}

\bibitem{DBLP:journals/dmtcs/DujmovicPW04}
Dujmovi{\'{c}}, V., P{\'{o}}r, A., Wood, D.R.: Track layouts of graphs.
  Discrete Mathematics {\&} Theoretical Computer Science  6(2),  497--522
  (2004), \url{http://dmtcs.episciences.org/315}

\bibitem{DBLP:journals/dmtcs/DujmovicW04}
Dujmovi{\'{c}}, V., Wood, D.R.: On linear layouts of graphs. Discrete
  Mathematics {\&} Theoretical Computer Science  6(2),  339--358 (2004),
  \url{http://dmtcs.episciences.org/317}

\bibitem{DBLP:conf/focs/Heath84}
Heath, L.S.: Embedding planar graphs in seven pages. In: {FOCS}. pp. 74--83.
  {IEEE} Computer Society (1984),
  \url{https://doi.org/10.1109/SFCS.1984.715903}

\bibitem{DBLP:journals/siamdm/HeathLR92}
Heath, L.S., Leighton, F.T., Rosenberg, A.L.: Comparing queues and stacks as
  mechanisms for laying out graphs. {SIAM} J. Discrete Math.  5(3),  398--412
  (1992), \url{https://doi.org/10.1137/0405031}

\bibitem{DBLP:journals/siamcomp/HeathR92}
Heath, L.S., Rosenberg, A.L.: Laying out graphs using queues. {SIAM} J. Comput.
   21(5),  927--958 (1992), \url{https://doi.org/10.1137/0221055}

\bibitem{DBLP:journals/dam/HortonPB00}
Horton, S.B., Parker, R.G., Borie, R.B.: On minimum cuts and the linear
  arrangement problem. Discret. Appl. Math.  103(1-3),  127--139 (2000),
  \url{https://doi.org/10.1016/S0166-218X(00)00173-6}

\bibitem{DBLP:journals/corr/abs-1210-8113}
Kratochv{\'{\i}}l, J., Vaner, M.: A note on planar partial 3-trees. CoRR
  abs/1210.8113 (2012), \url{http://arxiv.org/abs/1210.8113}

\bibitem{DBLP:journals/combinatorics/Wiechert17}
Wiechert, V.: On the queue-number of graphs with bounded tree-width. Electr. J.
  Comb.  24(1),  P1.65 (2017),
  \url{http://www.combinatorics.org/ojs/index.php/eljc/article/view/v24i1p65}

\bibitem{DBLP:conf/fsttcs/Wood02}
Wood, D.R.: Queue layouts, tree-width, and three-dimensional graph drawing. In:
  Agrawal, M., Seth, A. (eds.) {FST} {TCS}. LNCS, vol. 2556, pp. 348--359.
  Springer (2002), \url{https://doi.org/10.1007/3-540-36206-1\_31}

\bibitem{DBLP:journals/jcss/Yannakakis89}
Yannakakis, M.: Embedding planar graphs in four pages. J. Comput. Syst. Sci.
  38(1),  36--67 (1989), \url{https://doi.org/10.1016/0022-0000(89)90032-9}

\end{thebibliography}

\end{document}